\documentclass[letterpaper, 10 pt, conference]{ieeeconf}  % Comment this line out
\IEEEoverridecommandlockouts                              % This command is only
\title{\LARGE \bf
A PAC-Bayesian Framework for Optimal Control\\ with Stability Guarantees
}
\author{Mahrokh Ghoddousi Boroujeni,\textsuperscript{1} Clara Luc\'{i}a Galimberti,\textsuperscript{1} Andreas Krause,\textsuperscript{2} and Giancarlo Ferrari-Trecate\textsuperscript{1}
    \thanks{
        \textsuperscript{1} Institute of Mechanical Engineering, EPFL, Switzerland.
        \texttt{\{mahrokh.ghoddousiboroujeni, clara.galimberti, giancarlo.ferraritrecate\}@epfl.ch}.
    }
    \thanks{
        \textsuperscript{2} Department of Computer Science, ETH Zürich, Switzerland.
        \texttt{krausea@ethz.ch}.
    }
    \thanks{
        Research supported by the Swiss National Science Foundation under the NCCR Automation (grant agreement 51NF40\textunderscore80545).
    } 
}

\usepackage{cite}
\usepackage{amsmath,amssymb,amsfonts}
\usepackage{algorithmic}
\usepackage{graphicx}
\usepackage{textcomp}
\usepackage{bm}
\usepackage{hyperref}
\usepackage{siunitx}
\usepackage{mathtools}

\DeclareMathAlphabet\mathbfcal{OMS}{cmsy}{b}{n}
\newcommand{\R}{\mathbb{R}}
\newcommand{\N}{\mathbb{N}}
\renewcommand{\S}{\mathbb{S}}
\newcommand{\U}{\mathbb{U}}
\newcommand{\X}{\mathbb{X}}
\newcommand{\Y}{\mathbb{Y}}     % only used for supervised learning
\newcommand{\W}{\mathbb{W}}
\newcommand{\thetaset}{\Theta}
\newcommand{\D}{\mathcal{D}}
\renewcommand{\P}{\mathcal{P}}
\newcommand{\Q}{\mathcal{Q}}
\newcommand{\Emme}[0]{\mathcal{M}}

\renewcommand{\L}{\mathcal{L}}
\newcommand{\Lhat}{\hat{\mathcal{L}}}
\DeclareMathOperator*{\E}{\mathbb{E}}
\newcommand{\numrollouts}{s}

\DeclareMathOperator*{\argmin}{arg\,min}
\DeclareMathOperator*{\KL}{\mathrm{KL}}
\newtheorem{proposition}{Proposition}
\newtheorem{theorem}{Theorem}
\newtheorem{corollary}{Corollary}

\newtheorem{assumption}{Assumption}

\usepackage{xcolor}
\usepackage{soul}

\newcommand{\norm}[1]{\left\lVert#1\right\rVert}

\makeatletter
\def\endthebibliography{%
  \def\@noitemerr{\@latex@warning{Empty `thebibliography' environment}}%
  \endlist
}

\let\labelindent\relax
\usepackage[inline]{enumitem}
\newlist{mylist}{enumerate*}{1}
\setlist[mylist]{label=(\roman*)}

\usepackage{booktabs}

\begin{document}

\maketitle
\thispagestyle{empty}
\pagestyle{empty}

%%%%%%%%%%%%%%%%%%%%%%%%%%%%%%%%%%%%%%%%%%%%%%%%%%%%%%%%%%%%%%%%%%%%%%%%%%%%%%%%
\begin{abstract}
Stochastic Nonlinear Optimal Control (SNOC) involves minimizing a cost function that averages out the random uncertainties affecting the dynamics of nonlinear systems. For tractability reasons, this problem is typically addressed by minimizing an empirical cost, which represents the average cost across a finite dataset of sampled disturbances. However, this approach raises the challenge of quantifying the control performance against out-of-sample uncertainties. Particularly, in scenarios where the training dataset is small, SNOC policies are prone to overfitting, resulting in significant discrepancies between the empirical cost and the true cost, i.e., the average SNOC cost incurred during control deployment. Therefore, establishing generalization bounds on the true cost is crucial for ensuring reliability in real-world applications. In this paper, we introduce a novel approach that leverages PAC-Bayes theory to provide rigorous generalization bounds for SNOC. Based on these bounds, we propose a new method for designing optimal controllers, offering a principled way to incorporate prior knowledge into the synthesis process, which aids in improving the control policy and mitigating overfitting. Furthermore, by leveraging recent parametrizations of stabilizing controllers for nonlinear systems, our framework inherently ensures closed-loop stability. The effectiveness of our proposed method in incorporating prior knowledge and combating overfitting is shown by designing neural network controllers for tasks in cooperative robotics. 
\end{abstract}

%%%%%%%%%%%%%%%%%%%%%%%%%%%%%%%%%%%%%%%%%%%%%%%%%%%%%%%%%%%%%%%%%%%%%%%%%%%%%%%%
\section{Introduction}
\looseness -1 Stochastic Nonlinear Optimal Control (SNOC) involves minimizing a cost function that averages out stochastic uncertainties affecting the nonlinear system, such as process noise and disturbances. 
Beyond the linear quadratic Gaussian scenario, this problem is often intractable and is typically replaced by minimizing the average cost across a dataset of sampled uncertainties, namely, the empirical cost. 
However, this approximation
leads to the challenge of determining generalization bounds, i.e., to upper-bound the true (uncomputable) cost using its (computable) empirical version. 
A remarkably similar problem has been thoroughly explored in Machine Learning within the scope of Probably Approximately Correct (PAC)-Bayesian learning~\cite{userfriendly}.
Here, the generalization error of predictors drawn from an arbitrary posterior distribution is upper-bounded with a high probability, either in expectation or for a single draw. This upper bound comprises the empirical cost and a discrepancy measure between the posterior and a prior distribution.

PAC-Bayesian bounds have never been exploited for the analysis of SNOC problems, with the exception of~\cite{majumdar2021pac}, where it was employed in a different context, focusing on zero-shot generalization to new environments. 
Related contributions utilize PAC-Bayes for learning control policies in Markov Decision Processes~\cite{J1, J2}.  
However, these approaches assume discrete action or state spaces, which might not be well-suited for SNOC. 
One reason for the limited application of classic PAC-Bayesian bounds in control is that they require the use of aggregated predictors, resulting from weighted averaging of all predictors according to the posterior distribution.
In the SNOC context, such aggregation would require integrating actions from all possible controllers according to the posterior, a computationally challenging and impractical approach. Instead, we adopt a less commonly employed variant of PAC-Bayes bounds that apply to randomized predictors~\cite{userfriendly}, which are single predictors drawn from the posterior. This allows us to establish a generalization bound when deploying a \emph{single controller} sampled from the posterior, which is more practical.
The proposed PAC-Bayesian bound provides an alternative to existing methods for validating certificates of learned controllers, such as those based on mixed-integer programs or satisfiability-modulo-theory certificates~\cite{dawson2023safe, schwan2023stability}. 

\looseness -1 PAC-Bayesian bounds apply to control policies that are derived from the posterior distribution, either through integration or sampling. 
Ensuring stability when deploying the resulting controller is crucial in control systems. This requires that the posterior distribution assigns a non-zero weight only to controllers that stabilize the closed-loop system.

Recently, there has been a significant advancement 
in using neural networks in control loops~\cite{NeurSLS,furieri2024learning,wang2021learning,furieri2022distributed,gu2021recurrent,pauli2021offset,bonassi2022recurrent} or for providing stability certificates~\cite{abate2020formal,berkenkamp2018safe}.
In particular,\cite{NeurSLS,furieri2024learning,wang2021learning,furieri2022distributed} %These developments 
allow using broad families of nonlinear controllers based on neural networks to guarantee closed-loop stability while optimizing general cost functions. 
The flexibility in selecting the cost function enables promoting a broad spectrum of closed-loop specifications, ranging from safety to minimizing the L$_2$-gain, among others. 
In this work, we exploit
the Performance-Boosting Control architecture, an unconstrained parametrization of all and only stabilizing controllers for a given system~\cite{NeurSLS,furieri2024learning}.
Their implementation leverages Deep Neural Networks (DNNs) from~\cite{revay2023recurrent} 
which are flexible models capable of achieving complex tasks. 
Since~\cite{NeurSLS,furieri2024learning} guarantee closed-loop stability for any parameter choice, every posterior distribution over the induced space of parameters will exclusively assign non-zero probability to stabilizing controllers.

\subsection{Contributions}
\looseness -1 We develop a novel framework for establishing a generalization bound for SNOC cost based on randomized bounds in PAC-Bayes theory, a class of bounds not previously exploited in control design. 
Based on the derived bound, we provide a new control policy design algorithm, which involves sampling the controller from an inferred posterior distribution. 
We define this posterior distribution only over stabilizing controllers for a given system, enabling the deployment of our probabilistic algorithm while ensuring stability by design.
To our knowledge, this is the first PAC-Bayesian control design algorithm with stability guarantees.
Finally, we employ Stein Variational Gradient Descent (SVGD)~\cite{SVGD} to make sampling the controller from the posterior numerically tractable.
We demonstrate the feasibility of our method through two examples.
First, we showcase how our approach can incorporate prior coarse knowledge into control design using a toy example. Then, we demonstrate the potential of our approach in a cooperative robotic task.

\subsection{Notation}
$\N_0$ denotes $\N \cup \{0\}$. $\R^+$ ($\R^+_0$) stands for positive (non-negative) real numbers.
We denote the truncated sequence $(x_0,x_1,\ldots,x_{t})$ as $x_{t:0}$, where $x_\tau \in \mathbb{R}^n$ for all $\tau=0,\dots,t$. In the limit, as $t\rightarrow \infty$, this sequence is denoted as $x_{\infty:0}\coloneqq (x_0,x_1,\ldots)$.
The set of all sequences $x_{\infty:0}$ is denoted as $\ell^n$. 
Moreover,  $x_{\infty:0}$ belongs to $\ell_p^n \subset \ell^n$ 
for $p \in \mathbb{N}\cup\{+\infty\}$,  
if $\norm{x_{\infty:0}}_p = \left(\sum_{\tau=0}^\infty |x_\tau|^p\right)^{1/p} < +\infty$ for $p \in \mathbb{N}$, and $\operatorname{sup}_{\tau}|x_\tau|< +\infty$, for $p=\infty$,
where $|\cdot|$ denotes the Euclidean norm.
We represent causal operators over sequences with $A_{\infty:0} : \ell^n\rightarrow\ell^{n'}$, defined as $A_{\infty:0}(x_{\infty:0}) = (A_0(x_0), A_1(x_{1:0}), \dots, A_t(x_{t:0}), \dots )$.
Every instance $A_t(x_{t:0})$, for a fixed $t\in\N_0$, is a function in $\mathbb{R}^{n\times(t+1)} \rightarrow \mathbb{R}^{n'}$.
The operator $A_{\infty:0}$ is said to be $\ell_p$-stable if $A_{\infty:0}(x_{\infty:0}) \in \ell_p^{n'}$ for all $x_{\infty:0}\in\ell_p^n$.
Similar to sequences, we define truncated operators as 
$A_{t:0}(x_{t:0}) = (A_0(x_0), A_1(x_{1:0}), \dots, A_t(x_{t:0}))$.
Finally, the zero vector in $\mathbb{R}^n$ is denoted as $0_n$.

%%%%%%%%%%%%%%%%%%%%%%%%%%%%%%%%%%%%%%%%%%%%%%%%%%%%%%%%%%%%%%%%%%%%%%%%%%%%%%%%
\section{Stochastic nonlinear optimal control}
We consider nonlinear discrete-time time-varying systems described by the equation:
\begin{subequations}
\label{eq:system}
\begin{align}
    x_0 &= \Bar{x} + w_0, \label{eq:system_ini} \\
    x_t &= f_t(x_{t-1:0}, u_{t-1:0}) + w_t, \quad t=1,2,\ldots,\label{eq:system_step}
\end{align}
\end{subequations}
where $\Bar{x} \in \X \subset \R^n$ denotes the nominal initial state, $x_t \in \X \subset \R^n$ represents the state vector, $u_t \in \U\subset \R^m$ denotes the control input, and $w_t \in \W \subset \R^n$ stands for unknown process noise.
The uncertainty on $x_0$ is modeled through $w_0$.
The noise follows a potentially unknown distribution, $w_t \sim \D_t$.
We assume the nominal initial state, $\Bar{x}$, and the system dynamics, $f_t$, are known.
To control the system~\eqref{eq:system}, we employ a time-varying dynamical feedback controller parameterized by $\theta \in \thetaset \subset \R^d$, given by: \looseness-1
\begin{equation}
    \label{eq:control}
    u_t = -K_{t}^\theta (x_{t:0}), \quad t=0,1,\ldots.
\end{equation}
Given a system, a controller parameterized by $\theta$, and a sequence of noise samples, $w_{t:0}$, the closed-loop system~\eqref{eq:system}-\eqref{eq:control} yields unique state and input trajectories, $(x_{t:0}, u_{t:0}) =r^\theta_{t:0}(w_{t:0})$.
The deterministic mapping $r^\theta_{t:0} : \W^{t+1} \rightarrow (\X {\times} \U)^{t+1}$, is referred as the \emph{rollout map}, where its components are $r^\theta_t : w_{t:0}\mapsto (x_t,u_t)$.
For notation simplicity, we omit the dependence of $r^\theta_{t:0}$ on the nominal initial state, $\Bar{x}$, and the system dynamics, $f_{t}$.
Moreover, when clear from the context, we will drop the subscript and write $r^\theta$.
In this work, we assume that system~\eqref{eq:system} is $\ell_p$-stable or has been already prestabilized, i.e., the map $(w_{\infty:0},u_{\infty:0})\rightarrow x_{\infty:0}$ is $\ell_p$-stable. 
As in~\cite{NeurSLS}, the closed-loop system~\eqref{eq:system}-\eqref{eq:control} is called $\ell_p$-stable if the rollout map, $r^\theta_{\infty:0}$, is $\ell_p$-stable.

We consider the following SNOC problem over a finite horizon of $T\in \N$: 
\begin{subequations}\label{eq:SNOC}
\begin{align}
    \text{SNOC}:~~ \min_{\theta \in \thetaset}\quad& \L(\theta, \D_{T:0}), \label{eq:snoc_ob}\\
    \textit{s.t.}\quad& (x_t, u_t) = r^\theta_t(w_{t:0}), \quad t = 0, \ldots, T,  \label{eq:snoc_dynamics}\\
    &r^\theta_{\infty:0} \text{ is } \ell_p\text{-stable}, \label{eq:snoc_stability}
\end{align}
\end{subequations}
where,
\begin{align}
    \L(\theta, \D_{T:0}) &\coloneqq
        \E_{w_{T:0} \sim \D_{T:0}} L(\theta, w_{T:0}),\label{eq:true_cost}\\
    L(\theta, w_{T:0}) &\coloneqq 
        \frac{1}{T} \sum_{t=0}^{T}l(x_t,u_t). \label{eq:cost_L}
\end{align}
The objective is to synthesize a controller parameterized by $\theta$ that minimizes the cost $\L$ while adhering to the closed-loop dynamics and stability constraints imposed by~\eqref{eq:snoc_dynamics} and~\eqref{eq:snoc_stability}, respectively.
In~\eqref{eq:cost_L}, $l: \X \times \U \rightarrow \R_0^+$ denotes a stage cost function, which we assume to be piecewise differentiable and nonnegative for all $(x,u) \in \X \times \U$.
Given $\theta$ and $w_{T:0}$, $L$ defined in~\eqref{eq:cost_L} averages the stage cost over the horizon $T$, referred to as the \textit{finite-horizon (FH) cost}.
The cost $\L$ in~\eqref{eq:true_cost}, dubbed the \emph{true cost}, represents the expectation of the FH cost for any realization of noise sampled from $\D_{T:0}\coloneqq (\D_0, \ldots, \D_T)$.
By averaging out the random noise affecting the dynamics, the true cost, $\L$, provides a suitable criterion for optimizing the controller parameters, as in~\eqref{eq:SNOC}.

Computing the true cost is often prohibitive due to involving an expectation calculation or is infeasible due to unknown noise distributions, $\D_{T:0}$. Instead, SNOC methods typically resort to optimizing the \emph{empirical cost} based on observed data, $\S \coloneqq \bigl\{w_{T:0}^i\bigr\}_{i=1}^\numrollouts$, consisting of $\numrollouts \in \N$ noise sequences of length $T+1$ drawn from $\D_{T:0}$, indicated as $\S \sim \D^s_{T:0}$.
The empirical cost is defined as:\looseness-1
\begin{align}
\Lhat(\theta, \S) &\coloneqq 
\frac{1}{\numrollouts} \sum_{i=1}^{\numrollouts} L(\theta, w_{T:0}^i)
\label{eq:emp_cost},
\end{align}
averaging the FH cost of noise sequences in $\S$. 
Unlike $\L(\theta, \D_{T:0})$, the empirical cost can be computed and optimized to obtain the controller parameters. The resulting controller is referred to as the \emph{empirical controller}.
However, the empirical controller may exhibit significantly worse performance when confronted with out-of-sample noise sequences than those in the training set $\S$, a phenomenon corresponding to overfitting in supervised learning. 
To rigorously certify the performance of $\L$, we aim to upper-bound it using its empirical estimate, $\Lhat$, employing the PAC-Bayes framework, briefly introduced in the next section.
In Section~\ref{sec:controller_parametrization}, we delve into considerations about closed-loop stability.

\section{Generalization bounds}\label{sec:PAC}
In this section, we establish a generalization bound on the true cost and use it to propose a control design algorithm.

\subsection{PAC-Bayes generalization bound}
We leverage the PAC-Bayes framework that applies to learning algorithms that incorporate a prior distribution $\P$ with the observed data $\S$ to derive a posterior distribution $\Q$. 
Importantly, $\P$ and $\Q$ are not necessarily linked through Bayes' law. In the context of learning controllers, both $\P$ and $\Q$ are defined over the controller parameters, $\theta$.
In this study, we utilize PAC-Bayes bounds that necessitate the prior distribution to be chosen independently from the data and yield bounds for individual samples from the posterior.

We make the following assumption which facilitates deriving PAC-Bayesian bounds.
\begin{assumption}\label{assumption}
    The FH cost is upper-bounded by $C\in \R^+$, 
    i.e.,
    $L(\theta, w_{T:0}) \in [0, C)$, for all $\theta \in \thetaset$ and $w_{T:0} \in \W^{T+1}$.
\end{assumption}

Assumption~\ref{assumption} may not apply to a wide range of SNOC problems.
For instance, even a simple quadratic stage cost can lead to an unbounded $L$ if the closed-loop system is unstable. We propose to map any unbounded cost function to $[0, C)$. To achieve this, we choose the transformation below:
\begin{align}
    \Tilde{L}(x_{T:0},u_{T:0}) \coloneqq C \tanh(L(x_{T:0},u_{T:0})/\gamma), \label{eq:cost_trans}
\end{align}
where $\tanh$ represents the hyperbolic tangent function and $\gamma \in \R^+$ is a suitable constant.
This transformed cost grows almost linearly with the original cost between $0$ and $\gamma$,
and then smoothly saturates at the maximum value of $C$. 
In general, we expect that the stage cost, $l$, decreases over time. Therefore, the cost at the initial nominal state, $l(\Bar{x}, 0_m)$, is a reasonable selection for $\gamma$. 
This choice retains the transformed cost predominantly within the linear region rather than in the saturation region.
In the sequel, we apply~\eqref{eq:cost_trans} to any unbounded FH cost~\eqref{eq:cost_L} used in the SNOC problem~\eqref{eq:SNOC}, whenever necessary.

Next, we present our PAC-Bayes bound on the true cost,~$\L$.
\begin{theorem}[Adapted from Theorem~2.7 in~\cite{userfriendly}]
\label{theo:pac}
    Consider the noise distribution $\D_{T:0}$, a dataset $\S \sim \D_{T:0}^\numrollouts$, a prior distribution $\P$ independent of $\S$, and the FH cost $L$.
    Under Assumption~\ref{assumption}, for every $\lambda>0$, confidence level $\delta \in (0,1)$, posterior distribution $\Q$, and controller parameters $\theta \sim \Q$, the inequality,
    \begin{align}
        \L(\theta, \D_{T:0}) 
        \leq \Lhat(\theta, \S)  
        + \frac{1}{\lambda}\Bigl(
            \ln \frac{d\Q}{d\P}(\theta) + \ln \frac{1}{\delta}
        \Bigr)
        + \frac{\lambda C^2}{8 \numrollouts}, \label{eq:bound}
    \end{align}
    holds with probability at least $1-\delta$ over simultaneously sampling $\S \sim \D_{T:0}^\numrollouts$ and $\theta \sim \Q$. The true and empirical costs, $\L$ and $\Lhat$, are defined in~\eqref{eq:true_cost} and~\eqref{eq:emp_cost} respectively.
\end{theorem}

Theorem~\ref{theo:pac} provides an upper bound on the true cost of any $\theta$ sampled from $\Q$ based on its empirical cost, a discrepancy measure between the posterior and prior at $\theta$, and a term that is constant with respect to $\theta$.
Theorem~\ref{theo:pac} and its proof are derived by adapting Theorem~2.7 from~\cite{userfriendly}, originally formulated for supervised learning, to the SNOC framework. This adaptation involves establishing a correspondence between elements in supervised learning and elements in our problem, as summarized in Table~\ref{tab:sup-snoc}. While supervised learning deals with data samples comprising features and labels, in our context, we work with noise sequences. 
In supervised learning, a predictor, $h^{\theta}$, maps features to a label, whereas in SNOC, the rollout map, $r^\theta$, generates closed-loop state and input sequences. 
Moreover, in supervised learning, a loss function evaluates a predictor at a given sample point, while in SNOC, we use the FH cost to evaluate trajectories generated by the rollout map. 
By recognizing and aligning these correspondences, we can effectively apply PAC-Bayes bounds from supervised learning to the SNOC domain.
\begin{table}
\vspace{6pt}
\caption{Correspondence between supervised learning and SNOC}
\label{tab:sup-snoc}
\setlength{\tabcolsep}{1pt}
\begin{center}
\begin{tabular}{lllll}
\hline
\multicolumn{2}{c}{Supervised Learning} &  & \multicolumn{2}{c}{SNOC} \\ \cline{1-2} \cline{4-5} \\ 
Sample & $(x,y) \in \X {\times} \Y$             && Noise sequence & $w_{T:0} \in \W^{T+1}$           \\
Predictor & $h^{\theta}: \X {\xrightarrow[]{}} \Y$ && Rollout        & $r^\theta: \W^{T+1} {\xrightarrow[]{}} (\X{\times} \U)^{T+1}$ \\
Loss                & $l(\theta, x, y)$               && FH cost          &  $L(\theta, w_{T:0})$      \\
\hline
\end{tabular}
\end{center}
\end{table}

The majority of PAC-Bayes bounds in supervised learning focus on \textit{aggregated predictors}, which combine the output of all possible predictors according to the posterior distribution, $\E_{\theta \sim \Q} h^\theta(\cdot)$~\cite{userfriendly}. 
Consequently, these bounds involve the average of true and empirical losses with respect to sampling predictors from the posterior, denoted as $\E_{\theta \sim \Q} \L(\theta, \cdot)$ and $\E_{\theta \sim \Q} \Lhat(\theta, \cdot)$, respectively (see, for instance,~\cite{Catoni}). Adapting PAC-Bayes bounds for aggregated predictors to SNOC leads to obtaining bounds in which the control action is obtained by combining actions from all possible controllers sampled from the posterior, i.e., $u_t = - \E_{\theta \sim \\Q} K^{\theta}_t (x_{t:0})$. 
However, such averaging is generally impractical in control systems, where using a unique controller to obtain $u_t$ is preferred. Moreover, computing the expectation can be computationally prohibitive.
Therefore, Theorem~\ref{theo:pac} is derived by adapting a less commonly used family of PAC-Bayes bounds in supervised learning that considers randomized predictors, which obtain the output using a single predictor drawn from the posterior distribution. 
Consequently, Theorem~\ref{theo:pac} aligns well with the standard setup in control systems as it holds for a fixed sample of the controller parameters from the posterior, $\theta \sim \Q$.

\subsection{Optimal posterior}\label{subsec:qstar}
While the prior and posterior terms in Theorem~\ref{theo:pac} resemble Bayesian terminology, the posterior distribution, $\Q$, is not necessarily obtained through Bayes’ theorem. 
The flexibility in choosing $\Q$ can be exploited to achieve desired characteristics.
In particular, for a given $\theta$ sampled from $\Q$, it is desirable that $\Q$ is selected to minimize the upper bound on $\L$, which is the right-hand side of~\eqref{eq:bound}.
However, this choice is hampered by a circular dependency since $\Q$ must be fixed before sampling $\theta$.
To overcome this issue, we propose to select $\Q$ by minimizing the expected value of the upper bound over the draw of the controller parameters:
\begin{align}
    \Q^* \coloneqq& \argmin_{\Q} \E_{\theta \sim \Q}  \Bigl[\Lhat(\theta, \S) 
    + \frac{1}{\lambda} \Bigl( \ln \frac{d\Q}{d\P}(\theta)
    + \ln \frac{1}{\delta} \Bigr) 
    + \frac{\lambda C^2}{8 \numrollouts} \Bigr] \nonumber \\
    =& \argmin_{\Q} \E_{\theta \sim \Q}  \Bigl[\Lhat(\theta, \S) 
    + \frac{1}{\lambda} \ln \frac{d\Q}{d\P}(\theta)\Bigr] \nonumber\\
    =& \argmin_{\Q} \E_{\theta \sim \Q}  \Bigl[\Lhat(\theta, \S) \Bigr] +
    \frac{1}{\lambda} \KL\bigl(\Q \Vert \P \bigr) 
    ,\label{eq:def_qstar}
\end{align}
where $\KL$ stands for the Kullback–Leibler divergence.
This expression involves a trade-off, modulated by $\lambda$, between minimizing the expected empirical cost and enhancing the similarity of the posterior to the prior. This similarity is captured through the KL divergence term.
The prior can be selected such that the similarity term promotes the integration of existing knowledge into the learning procedure or regulates the posterior complexity to avoid overfitting. 
Examples of such priors are provided in Section~\ref{sec:experiments}.
In extreme cases, $\Q^*$ becomes a degenerate distribution at the minimizer(s) of $\Lhat$ as $\lambda \xrightarrow[]{} \infty$, and $\Q^* \xrightarrow[]{} \P$ as $\lambda \xrightarrow[]{} 0$. 
Below, we present first the closed-form solution for $\Q^*$ and then an approach to select $\lambda$.

\begin{corollary}[Lemma~1.1.3 in~\cite{Catoni}]\label{corol:qstar}
    The minimizer of~\eqref{eq:def_qstar} is the \textit{Gibbs} distribution, characterized by:
    \begin{align} 
    \Q^*(\theta) 
    &= \P(\theta) \,
    e^{-\lambda \Lhat(\theta, \S)}
    / Z_\lambda(\P,\S) \label{eq:qstar},\\
    Z_\lambda(\P,\S) &\coloneqq \E_{\theta \sim \P} 
    e^{-\lambda \Lhat(\theta, \S)}, \label{eq:Z}
    \end{align}
    where $Z_\lambda(\P,\S)$ normalizes the distribution $\Q^*$. 
    \looseness -1
\end{corollary}
We refer to $\Q^*$ as the \textit{optimal posterior} distribution.
By substituting~\eqref{eq:qstar} into~\eqref{eq:bound}, one obtains:
\begin{align}
    \L(\theta, \D_{T:0}) 
    \leq
    \frac{1}{\lambda} \Bigl( \ln \frac{1}{\delta}
    - \ln Z_\lambda(\P,\S) \Bigr)
    + \frac{\lambda C^2}{8 \numrollouts} ,
    \label{eq:bound_qstar}
\end{align}
holding with probability at least $1-\delta$ over simultaneously sampling $\S \sim \D_{T:0}^\numrollouts$ and $\theta \sim \Q$.
The term $- 1/\lambda \ln Z_\lambda(\P,\S)$ always falls within the range $[0, C)$ and decreases when the prior $\P$ assigns higher probability to controller parameters with lower empirical costs for $\S$, as evident from~\eqref{eq:Z}. 
Therefore, the bound tightens when the prior $\P$ is ``more aligned'' with $\S$. Recall, however, that the prior $\P$ is chosen independently from $\S$.
The simplified bound~\eqref{eq:bound_qstar} applies when using the optimal posterior distribution, $\Q^*$, hence eliminating the explicit dependence on $\Q$, and is tighter than the generic bound in~\eqref{eq:bound}.

The bound~\eqref{eq:bound_qstar} is more valuable when its right-hand side is smaller. 
This can be achieved by tuning $\lambda>0$, which according to Theorem~\ref{theo:pac}, can be freely selected.
However, computing the term $-1/\lambda \ln Z_\lambda(\P, \S)$ in~\eqref{eq:bound_qstar} is often challenging as it involves calculating an expectation, as per~\eqref{eq:Z}.
Under Assumption~\ref{assumption}, this term is in $[0, C)$. Therefore, we replace $-1/\lambda \ln Z_\lambda(\P, \S)$ with $C$ to obtain an upper bound on~\eqref{eq:bound_qstar}. Below, we propose selecting $\lambda$ based on this looser upper bound. 
\begin{proposition}
    The optimal $\lambda>0$ that minimizes the upper bound on the right-hand side of~\eqref{eq:bound_qstar} obtained by replacing $-1/\lambda \ln Z_\lambda(\P, \S)$ with $C$, is given by:
    \begin{align}
        \lambda^* = \sqrt{8 \numrollouts \ln (1/\delta)} / C. \label{eq:lambdastar}
    \end{align}
\end{proposition}

Given a prior $\P$ and a dataset $\S$,~\eqref{eq:qstar} and~\eqref{eq:lambdastar} completely specify a posterior distribution, requiring no additional free parameters.  
The resulting distribution achieves a balance between fitting the data and incorporating prior knowledge.

\section{Learning stabilizing controllers}
\label{sec:controller_parametrization}
In order to guarantee closed-loop stability, we leverage the Performance-Boosting Control design approach of~\cite{NeurSLS,furieri2024learning}.
In general, stability can be imposed in an SNOC problem~\eqref{eq:SNOC} in two ways.
The first one is to resort to 
constrained optimization approaches that ensure global or local stability by enforcing appropriate Lyapunov-like inequalities during optimization~\cite{berkenkamp2018safe,gu2021recurrent,pauli2021offset}.
However, enforcing such constraints often becomes a computational bottleneck in complex applications.
This issue has been highlighted in~\cite{dawson2023safe}, discussing that popular stability verification methods are feasible for controllers with an order of $10^2$ parameters.
The second possibility is to use
unconstrained optimization approaches searching within classes of control policies with built-in stability guarantees~\cite{wang2021learning,furieri2022distributed, NeurSLS, furieri2024learning}.
These methods allow learning controllers through optimization algorithms based on gradient descent, without imposing further constraints to ensure stability.
Additionally, unconstrained approaches usually have a much lower computational burden than constrained ones~\cite{NeurSLS, furieri2024learning}.
In Section~\ref{subsec:numerical_SVGD}, we show that unconstrained approaches also enable the use of algorithms for sampling the posterior, such as SVGD.
Moreover, in Section~\ref{subsec:experiments_robots}, we demonstrate the scalability of these approaches by training a stabilizing controller with more than $10^4$ parameters.

In this work, we follow~\cite{NeurSLS,furieri2024learning}, which propose
an \emph{unconstrained parametrization} of all and only stabilizing nonlinear control policies for a nonlinear time-varying system, as per~\eqref{eq:system}, 
exploiting the internal model control framework~\cite{Economou_Morari_nlIMC_1986}. 
This method is based on choosing a unique free operator, denoted as $\Emme_{\infty:0}$.
The main result is that, assuming $\ell_p$-stability of system~\eqref{eq:system}, the closed-loop map is $\ell_p$-stable if and only if the operator $\Emme_{\infty:0}$ is also $\ell_p$-stable.%
\footnote{For ease of reading, we neglect the base controller in~\cite{NeurSLS}; assuming that, if needed, it is already included in the system dynamics $f_t$.}
The control action $u_t$, for any $t \in \N_0$, can be obtained as:
\begin{subequations}
\label{eq:neurSLS}
\begin{align}
    \hat{w}_t &= x_t - f_t(x_{t-1:0},u_{t-1:0}), \label{eq:neurSLS_w}\\
    u_t &= \Emme_t(\hat{w}_{t:0}). \label{eq:neurSLS_output}
\end{align}
\end{subequations}
The noise is reconstructed in~\eqref{eq:neurSLS_w} using the knowledge of the system dynamics, and then, $u_t$ is calculated though~\eqref{eq:neurSLS_output}.
Note that~\eqref{eq:neurSLS} is a dynamical feedback controller, as in~\eqref{eq:control}, with input $x_{t:0}$, output $u_t$, and internal states $\hat{w}_{t:0}$ and $u_{t:0}$.

Implementing~\eqref{eq:neurSLS} requires parameterizing $\Emme_{\infty:0}$. This is addressed in~\cite{NeurSLS,furieri2024learning} by adopting the Recurrent Equilibrium Networks (RENs) of~\cite{revay2023recurrent}, which allows representing a broad set of $\ell_2$-stable operators. 
Below, we briefly summarize this approach and refer the reader to~\cite{revay2023recurrent} for a more comprehensive explanation.
The resulting parametrized operator, $\Emme^\theta_{\infty:0}$, is a dynamical system, given by:
\begin{equation}\label{eq:REN}
    \begin{bmatrix}
    \xi_{t+1} \\
    \zeta_t \\
    u_t
    \end{bmatrix} 
	=
    {\Omega(\theta)}
    \begin{bmatrix}
    \xi_{t} \\
    \sigma(\zeta_t) \\ 
    \hat{w}_{t} 
    \end{bmatrix}
    \,,
    \quad 
    \xi_0 = 0_{n_\xi}
    \,,
\end{equation}
where $\xi_t\in\R^{n_\xi}$ is the internal state and $\zeta_t\in\R^{n_\zeta}$ is an internal variable. Recall that $\theta$, $\hat{w}_{t}$, and $u_t$ are the parameters, inputs, and outputs, respectively. The activation function $\sigma:\R\rightarrow\R$ is applied element-wise and it must be piece-wise differentiable with first derivatives restricted to the interval $[0, 1]$.
The mapping $\Omega: \R^d \rightarrow \R^{(n_\xi + n_\zeta + m)\times(n_\xi + n_\zeta + n)}$ ensures that $\zeta_t$ has a unique value for every $t\in\N_0$, and that~\eqref{eq:REN} is an $\ell_2$-stable operator for every $\theta \in \R^d$, i.e, one has $\thetaset=\R^d$. 
We also highlight that RENs can embed arbitrarily deep NNs by suitably selecting the structure of the matrix $\Omega(\theta)$~\cite{revay2023recurrent}. 
Overall,~\eqref{eq:neurSLS_w}-\eqref{eq:REN} specify a DNN-based controller which ensures closed-loop stability for any choice of parameters $\theta \in \R^d$.

\section{Practical implementation}
\subsection{Estimating the upper bound}\label{subsec:numerical_Z} 
In~\eqref{eq:bound_qstar}, we have derived an upper bound on the true cost when employing the optimal posterior. 
However, explicitly computing this upper bound is generally challenging due to the involvement of the $Z_\lambda$ term, as given by~\eqref{eq:Z}, which requires integrating over the controller parameters sampled from the prior. Consequently, we resort to bounding this term by its empirical counterpart.
\begin{proposition}\label{prop:num_ub}
Consider the setup of Theorem~\ref{theo:pac} when using the optimal posterior, $\Q^*$, given by~\eqref{eq:qstar}.
Let $\theta_1, \ldots, \theta_{n_p}$ be $n_p$ independent samples from the prior, $\P$. 
For every confidence level $\hat{\delta}\in (0,1)$ and any $n_p \geq (e^{\lambda C} - 1)^2 \ln(1/\hat{\delta})/2$, the inequality,
\begin{align*}
    \L(&\theta, \D_{T:0}) 
    \leq
    \frac{1}{\lambda} \ln \frac{1}{\delta} + \frac{\lambda C^2}{8 s}\nonumber \\
    &- \frac{1}{\lambda} \ln \left( 
        \frac{1}{n_p} \sum_{i=0}^{n_p} e^{- \lambda \Lhat(\theta_i, \S)} - 
        (1-e^{-\lambda C}) \sqrt{\frac{\ln(1/\hat{\delta})}{2 n_p}} 
    \right)
    ,
\end{align*}
holds with probability at least $(1-\delta)(1-\hat{\delta})$ over simultaneously sampling $\S \sim \D_{T:0}^\numrollouts$, $\theta \sim \Q$, and $\theta_1, \ldots, \theta_{n_p} \sim \P^{n_p}$.
\end{proposition}
\begin{proof}
Under Assumption~\ref{assumption}, $e^{-\lambda \Lhat(\theta, \S)}  \in [e^{-\lambda C}, 1]$ for any $\theta$. This allows applying Hoeffding's inequality~\cite{concentration} to lower-bound $Z_\lambda$:
\begin{align}
    Z_\lambda(\P, \S) \geq \frac{1}{n_p} \sum_{i=0}^{n_p} e^{- \lambda \Lhat(\theta_i, \S)} - 
    (1-e^{-\lambda C})\sqrt{\frac{\ln(1/\hat{\delta})}{2n_p}},\label{eq:proof_empub}
\end{align}
holding with probability at least $(1-\hat{\delta})$ over simultaneously sampling $\theta_1, \ldots, \theta_{n_p} \sim \P^{n_p}$.
Since $n_p \geq (e^{\lambda C} - 1)^2 \ln(1/\hat{\delta})/2$, the right-hand side of~\eqref{eq:proof_empub} is positive, enabling us to substitute~\eqref{eq:proof_empub} into~\eqref{eq:bound_qstar}.
Since $\theta$, $\S$, and $\theta_1, \ldots, \theta_{n_p}$ are independent random variables, the probability that both inequalities~\eqref{eq:bound_qstar} and~\eqref{eq:proof_empub} hold simultaneously is the multiplication of the probability that each inequality holds individually, i.e., $(1-\delta)(1-\hat{\delta})$.
\end{proof}
Proposition~\ref{prop:num_ub} leverages sampled controllers from the prior. 
For common prior distributions with known cumulative density functions, sampling from the prior is straightforward. Hence, the estimated bound can be easily computed.
Although the resulting bound is less tight than the original one in~\eqref{eq:bound_qstar}, the additional conservatism diminishes as $n_p$ grows.

% ----------------- SVGD -----------------
\subsection{Sampling from the posterior}
\label{subsec:numerical_SVGD}
The optimal posterior distribution, $\Q^*$, can be computed according to~\eqref{eq:qstar} up to the normalization constant $Z_\lambda(\P, \S)$, as discussed in Section~\ref{subsec:numerical_Z}. 
Direct sampling from such distributions that do not integrate to unity, called improper distributions, is challenging in general. Although one could use Markov chain Monte Carlo methods for sampling from an improper distribution, this approach tends to be computationally slow, especially for large datasets~\cite{green2015bayesian}.
Instead, we employ a Variational Bayes method, as proposed in~\cite{Alquier2016variational}, which allows approximating $Q^*$ by a tractable distribution and then sampling from this approximated distribution. 
Specifically, we use Stein Variational Gradient Descent (SVGD)~\cite{SVGD} which approximates $\Q^*$ using a set of $k\in\N$ particles, denoted as $\{\phi_1, \cdots, \phi_k\}$.
In our context, each particle $\phi_\kappa$ for $\kappa=1, \ldots, k$ corresponds to the parameters of a controller, i.e.,  $\phi_\kappa \in \thetaset$. SVGD initializes these particles by sampling them from the prior and then iteratively updates them aiming that their density in the parameter space $\thetaset$ matches the probability mass of $\Q^*$.
SVGD relies on gradient descent for updating the particles, rendering it well-suited for integration with the unconstrained parametrization outlined in Section~\ref{sec:controller_parametrization}.
After training, $\Q^*$ is approximated as a uniform distribution over the trained particles.
Hence, sampling from the approximated distribution corresponds to randomly selecting one of the particles with equal probability. The application of SVGD to approximate Gibbs distributions arising in PAC-Bayesian methods has been found useful in previous work, such as~\cite{pacoh, pacpfl}.
For more details and the precise implementation of SVGD, we refer to~\cite{SVGD}.

\section{Experiments}
\label{sec:experiments}
We conduct experiments using two systems: 
\textit{i)} a simple Linear Time-Invariant (LTI) system, and 
\textit{ii)} a system of two planar robots navigating to prespecified locations while avoiding collisions.
To train our algorithm and compute the upper bound, we utilize $\Q^*$ and $\lambda^*$ as per~\eqref{eq:qstar}, and~\eqref{eq:lambdastar} and apply the transformation~\eqref{eq:cost_trans} with $C=1$.
We compare our algorithm against the empirical controller. To maintain the standard setting, the empirical controller is trained without applying the transformation~\eqref{eq:cost_trans}.
The code to reproduce our experiments can be found at \url{https://www.github.com/DecodEPFL/PAC-SNOC}.

% ------ SCALAR ------
\subsection{LTI system}
\paragraph*{\textbf{Setup}}
We use a scalar stable LTI system, given by:
\begin{align*}
    x_0 &= \Bar{x} + w_0,\\ 
    x_{t} &= a x_{t-1} + b u_{t-1} + w_t , \quad t=1,2, \ldots,
\end{align*}
with $a=0.8$, $b=0.1$, and $\Bar{x} = 2$.
This simple setup will allow us to thoroughly illustrate our control design method and the derived PAC bounds.
At each time step, the noise $w_t$ is drawn from a Gaussian distribution with mean $\mu_w = 0.3$ and variance $\sigma_w^2 = 0.09$, i.e., $w_t \sim \mathcal{N}(0.3, 0.09)$.
We assume the noise distribution is unknown. Instead, we possess a dataset of $\numrollouts$ noise sequences over a horizon of $T=10$, and
we vary $\numrollouts$ to analyze its impact.

The stage cost function is quadratic and defined as $l_{LQ}(x_t, u_t) = 5 x_t^2 + 0.003 u_t^2$.
We utilize the affine state feedback controller $u_t = -(k x_t + \beta)$, where $k$ and $\beta$ represent the controller gain and bias. We set $\theta = [k, \beta]^\top$.
Ensuring closed-loop stability for this system corresponds to imposing $k \in (-2, 18)$. This constraint can be readily enforced during optimization by projecting the gain, $k$, into this set during training, obviating the need for complex techniques to guarantee stability, such as~\eqref{eq:neurSLS}.

\paragraph*{\textbf{Baselines}}
Our methodology is compared against two alternative approaches. First, an \textit{empirical} controller obtained by minimizing $\Lhat(\theta, \S)$ defined in~\eqref{eq:emp_cost}.
Second, a \textit{benchmark} (ideal) controller that assumes precise knowledge of the noise mean value, which
allows setting $\beta=\mu_w/b$. 
The benchmark controller's gain, $k$, is optimized by minimizing the empirical cost over a large dataset comprising $1024$ noise sequences. While the benchmark controller demonstrates strong performance, its practicality is limited by its dependence on a substantial amount of data and accurate knowledge of $\mu_w$. Nevertheless, it is involved in this experiment as the best-performing competitor. 

\paragraph*{\textbf{Prior distribution}}
The prior distribution for $\theta = [k, \beta]^\top$ must be chosen independently from the data, $\S$. We opt for a prior that is independent for $k$ and $\beta$. Specifically, the prior for $k$ is a Gaussian distribution centered at the infinite horizon LQR (IH-LQR) gain---which can be obtained using the stage cost $l_{LQ}$---and with a variance of $1.0$. Although the IH-LQR may not be optimal for our scenario with a finite horizon of $T=10$, it is a reasonable data-independent choice.
The choice of the prior distribution for $\beta$ can incorporate known information about the expected value of the noise, $\mu_w$.
We explore two scenarios: 
In the first case, we assume to only know that $\mu_w$ lies in the interval $[-0.5, 0.5]$ and we set the prior for $\beta$ as the uniform distribution over $[-0.5/b, 0.5/b]$, i.e., $\mathcal{U}(-5, 5)$.
In the second case, $\mu_w$ is known with some uncertainty, and we set the prior distribution for $\beta$ as a Gaussian distribution with a mean of $\mu_w/b=3$ and a variance of $1.5^2$.
These prior distributions are called $\P_\mathcal{U}$ and $\P_\mathcal{N}$, respectively.

\paragraph*{\textbf{Optimal posterior distribution}}
As both $\P$ and $\Q^*$ are bivariate distributions, we approximate their Probability Density Function (PDF) by discretizing a region encompassing over $95\%$ of the probability mass.
In Figure~\ref{fig:Gibbs_grid}, we depict the discretized 
$\P$ (left column) and $\Q^*$ when using $\numrollouts=8$ (middle column) and $\numrollouts=512$ (right column), and setting $\delta=0.2$. The top and bottom rows correspond to $\P_\mathcal{U}$ and $\P_\mathcal{N}$, respectively, and the color represents the approximated PDF. 
Additionally, the empirical and benchmark controllers are marked on the plots.

For $\numrollouts=8$, both posterior distributions associated with $\P_\mathcal{U}$ and $\P_\mathcal{N}$ are concentrated near the benchmark controller, whereas the empirical controller is positioned further away. 
This observation suggests that controllers sampled from the posterior are likely to more closely resemble the benchmark controller, potentially resulting in better out-of-sample performance than the empirical controller.
As the number of data sequences increases to $\numrollouts=512$, both posterior distributions become more tightly concentrated around the benchmark controller. As expected, the empirical distribution also improves, getting closer to the benchmark.
Comparing the rows of Figure~\ref{fig:Gibbs_grid} reveals that the more informative prior $\P_\mathcal{N}$ (second row), which reflects additional knowledge about $\mu_w$, yields a sharper posterior distribution. This observation suggests that our algorithm effectively integrates existing knowledge into the control design algorithm.

\begin{figure}
  \centering
  \vspace{1pt}
  \includegraphics[width=0.96\linewidth]{figs/N_biased_Gaussian_biased_wide_contour_T10_S8_512.pdf}
  \caption{Discretized PDF for the prior distributions (left) and the optimal posterior distributions with $\numrollouts=8$ (middle) and $\numrollouts=512$ (right). The top and bottom rows correspond to $\P_\mathcal{U}$ and $\P_\mathcal{N}$, respectively.
  Horizontal and vertical axes in each plot represent $\beta$ and $k$, while color indicates the PDF.
  The empirical and benchmark controllers are marked.}
  \label{fig:Gibbs_grid}
\end{figure}

\paragraph*{\textbf{Validity and tightness of the bound for $\Q^*$ and $\lambda^*$}}
A desirable upper bound ensures that the true cost of sampled controllers is below the upper bound in at least $1-\delta$ cases while avoiding being overly loose.
Below, we assess the quality of our bound~\eqref{eq:bound_qstar}.
We calculate the term $Z_\lambda$ on the right-hand side using the discretized prior distributions illustrated in Figure~\ref{fig:Gibbs_grid}.
For the left-hand side, that is the true cost of a sampled $\theta$, direct computation is intractable. 
Hence, we approximate it for each sampled $\theta$ by evaluating the empirical cost using a large dataset of $1024$ noise sequences, each of length $T{+}1{=}11$, sampled independently from $\mathcal{N}(\mu_w, \sigma_w^2)^{11}$. This approximation's accuracy is justified by the law of large numbers.
We investigate the influence of the training dataset size ($\numrollouts$), the prior distribution ($\P$), and the confidence level ($\delta$) on both the upper bound and the true cost. We recalculate $\Q^*$ and $\lambda^*$ and sample $\theta$ from the resulting posterior distribution for each scenario.

Figure~\ref{fig:ub} illustrates the true cost, $\L$, and the upper bound as per~\eqref{eq:bound_qstar} on the vertical axis for various configurations. The colors correspond to different choices for $\delta$ and the prior distribution, while the size of the training set, $\numrollouts$, is indicated on the $x$-axis. For each configuration, the true cost is estimated for $10$ different $\theta$ vectors sampled from $\Q^*$ through exhaustive test data sampling. The true costs of these sampled points are represented by circles aligned on the same vertical line. Note that the upper bound and the true cost in this figure are calculated after the transformation~\eqref{eq:cost_trans} has been applied to the FH cost to upper-bound it to $C=1$. Therefore, all values on the vertical axis are less than $1$.
\begin{figure}
  \centering
  \includegraphics[width=0.99\linewidth]{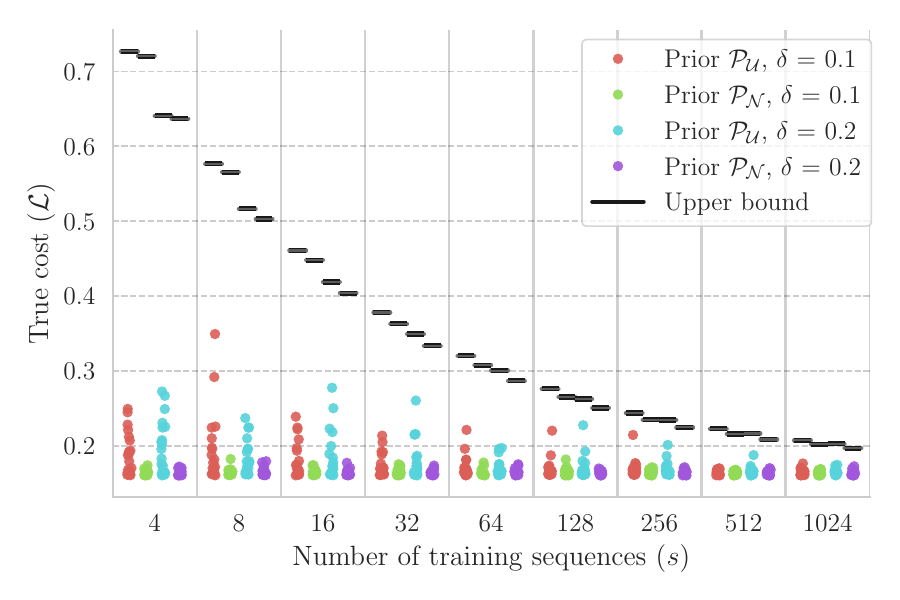}
  \caption{Comparison of the true cost, $\L$, and the upper bound~\eqref{eq:bound_qstar} for various configurations as a function of $\numrollouts$.
  Colors denote $\delta$ and prior distribution choices. The true cost in each setup is approximated for $10$ vectors $\theta$ sampled from $\Q^*$, shown as vertically aligned circles.}
  \label{fig:ub}
\end{figure}

In Figure~\ref{fig:ub}, we observe that the true cost consistently remains below the calculated upper bounds, indicating that they are all valid. When $\delta$ and $\numrollouts$ are fixed, the upper bound corresponding to $\P_\mathcal{N}$ is tighter than that for $\P_\mathcal{U}$. This observation aligns with our discussion in Section~\ref{subsec:qstar} that a more informative prior tightens the upper bound. 
Conversely, with a fixed prior and a fixed $\numrollouts$, increasing $\delta$ leads to a tighter upper bound, as it allows a higher probability of violating the bound. Moreover, as $\numrollouts$ increases, the upper bound uniformly tightens. A theoretical analysis of the impact of $\numrollouts$ on the upper bound will be developed in future research.

% ------ ROBOTS X ------
\subsection{Navigating planar robots}\label{subsec:experiments_robots}

\paragraph*{\textbf{Setup}}
We consider two point-mass robots moving in the Cartesian plane 
whose goal is to coordinately pass through a narrow valley and reach a target position 
while avoiding obstacles and collisions between them---see Figure~\ref{fig:corridor}.
Each robot is modeled as a double integrator subject to nonlinear drag forces.
The overall system state is denoted with $x\in\R^8$ and consists of the positions and velocities of the robots in the Cartesian coordinates. The inputs are the forces over the robots in each direction, represented by $u\in\R^4$.
Each robot is prestabilized by a simple proportional controller that takes it to the desired end position, possibly resulting in collisions and poor performance.
The system dynamics is subject to noise $w_t$ as in~\eqref{eq:system}.
We consider that the noise only affects the initial condition, i.e., $w_t=0$ for $t>0$.
Additionally, $w_0$ is drawn from a zero-mean Gaussian distribution with variance  $\sigma_w^2 = 0.2^2$, 
which we assume to be unknown. 
The dataset consists of $\numrollouts=30$  noise sequences over a horizon of $T=100$, where only the first elements are nonzero.

The stage cost is the sum of a quadratic function plus collision and obstacle avoidance regularizations, defined as:
\begin{equation*}
    l(x_t,u_t) = (\Delta x_t)^\top Q (\Delta x_t) + u_t^\top R u_t 
    + l_{ca}(d_t) + l_{oa}(x_t),
\end{equation*}
where $\Delta x_t = x_t - x_{\text{target}}$, $x_{\text{target}}$ is the desired state with the target positions and zero velocities, and $d_t\in\R_0^+$ is the distance between the robots at time $t$. 
The collision avoidance term is zero if the robots are further than a safe distance $D>0$. 
If $d_t<D$, one has $l_{ca}(d_t) = (d_t+\nu)^{-2}$, where $\nu>0$ is a small positive constant.
The obstacle avoidance term, $l_{oa}$, is a function of the distance between the robots and each obstacle.
This experiment is taken from~\cite{NeurSLS}, where the model and the stage cost function are further detailed.

\paragraph*{\textbf{Controller design}}
The controller is given by~\eqref{eq:neurSLS} and~\eqref{eq:REN}, ensuring that the closed-loop system remains stable as discussed in Section~\ref{sec:controller_parametrization}.
We set $n_\xi = 8$ and $n_\zeta = 8$, which leads to $d = 864$ parameters in total.
This controller architecture is trained once using the empirical approach which minimizes $\Lhat(\theta, \S)$, and once using our approach. 
In our method, we set $\delta = 0.1$ to ensure the bound~\eqref{eq:bound_qstar} holds with a probability of at least $0.9$. Additionally, we utilize a zero-mean spherical Gaussian prior with variance $7^2$, which imposes an L$_2$-regularization on the parameters $\theta$ of the REN~\eqref{eq:REN}. This choice is motivated by the fact that L$_2$-regularization of DNN weights helps mitigate overfitting in supervised learning. The variance of the prior is fine-tuned through cross-validation.

\paragraph*{\textbf{Results}}
Similar to the previous example, we approximate the controllers' true cost using a large test dataset of $500$ new noise trajectories.
To have a fair comparison, the test cost is calculated without transformation~\eqref{eq:cost_trans}.
Table~\ref{tab:robots} compares the cost and percentage of trajectories in which the robots collide, for the training and test datasets.
The empirical controller results in a marginally lower cost within the training dataset, whereas our approach exhibits superior performance in the test dataset. 
Both controllers avoid collisions in the training data but result in some collisions in the test data. This is because collision avoidance is only promoted through the stage cost function but not formally ensured.
Nevertheless, our method results in fewer collisions in the test data.
These observations highlight that the empirical controller tends to overfit the training samples, while our controller generalizes better to out-of-sample data.
Training the empirical controller and our controller took $212$ and $129$ minutes, respectively, using an AMD Ryzen Threadripper 3990X 64-Core Processor.

In Figure~\ref{fig:corridor}, each plot illustrates three trajectories starting from three initial conditions outside the training dataset in three cases: (\textit{a}) for the prestabilized system without an external controller, (\textit{b})-(\textit{c}) using the empirical controller, and (\textit{d})-(\textit{f}) employing the controller obtained by our approach. 
The utilized controllers for generating these plots are the same as those evaluated in Table~\ref{tab:robots}.
All trajectories are computed until $t=400$ to illustrate that stability is preserved.
As observed in the plot (\textit{a}), the prestabilizing controller leads to collisions between robots and oscillations around the target position, indicating poor performance.
Conversely, both of the employed controllers avoid collisions and achieve optimized trajectories in the plots (\textit{b})-(\textit{f}).
Animations are available in our \href{https://www.github.com/DecodEPFL/PAC-SNOC}{Github repository}.

\begin{table}
\caption{Comparison between the trained empirical controller and a sampled controller using our training approach.}
\label{tab:robots}
\begin{center}
\begin{tabular}{ccccc}
\toprule
\multicolumn{1}{c}{} & \multicolumn{2}{c}{\textbf{Cost}} & \multicolumn{2}{c}{\textbf{Collisions}} \\
\cmidrule(rl){2-3} \cmidrule(rl){4-5}
\textbf{Controller} & {Train} & {Test} & {Train} & {Test} \\
\midrule
Empirical & $\mathbf{21.79}$ & $23.14$ & $0.0\%$ & $5.8\%$ \\
Ours & $21.86$ & $\mathbf{22.29}$ & $0.0\%$ & $\mathbf{5.0\%}$ \\
\bottomrule
\end{tabular}
\end{center}
\end{table}

\begin{figure}
	\centering
	\begin{minipage}{0.3\linewidth}
        \includegraphics[width=\linewidth]{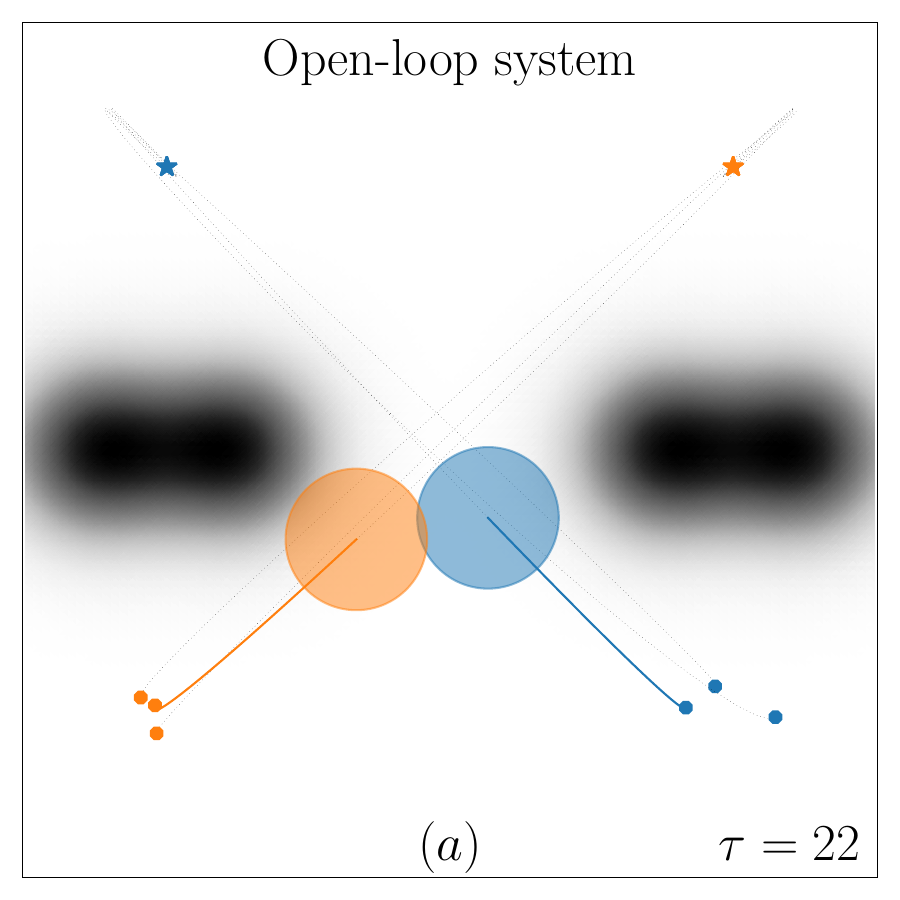}
	    \includegraphics[width=\linewidth]{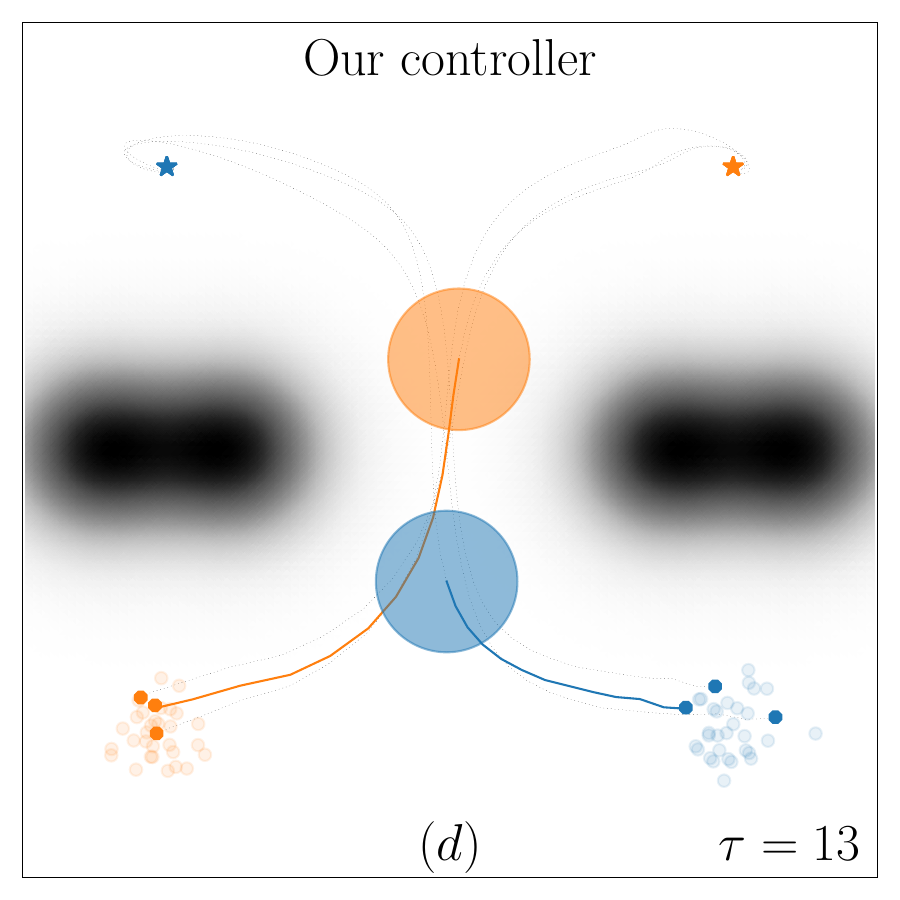}
	\end{minipage}%
	\begin{minipage}{0.3\linewidth}
        \includegraphics[width=\linewidth]{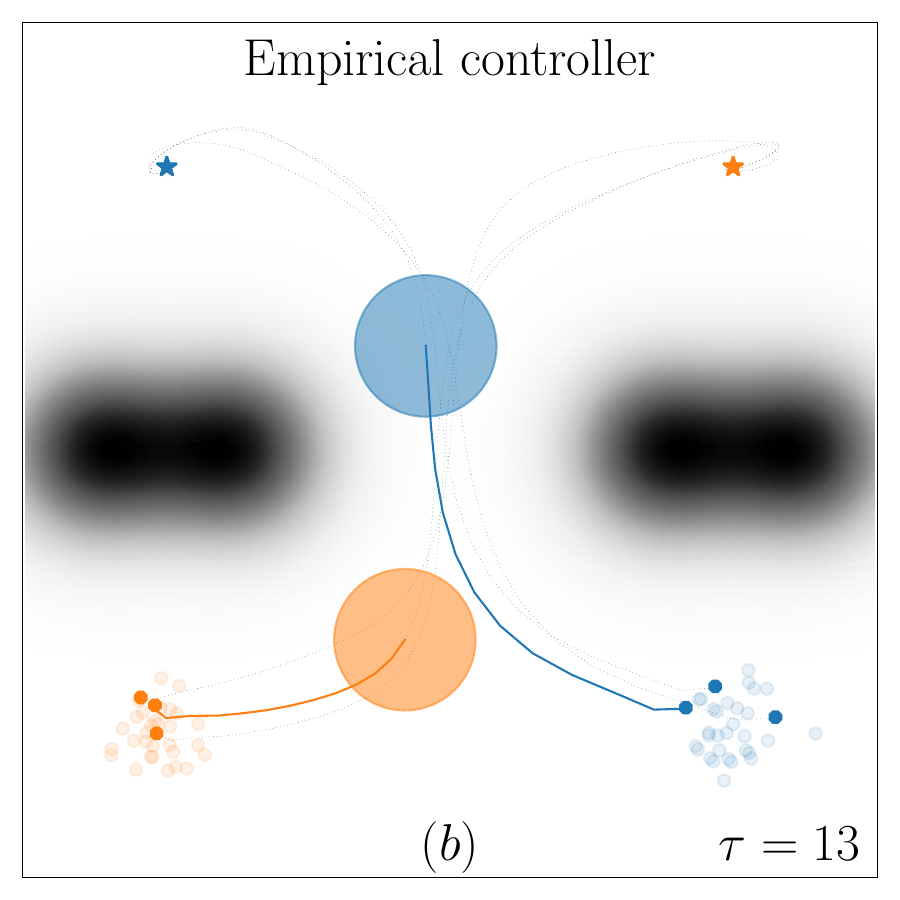}
	    \includegraphics[width=\linewidth]{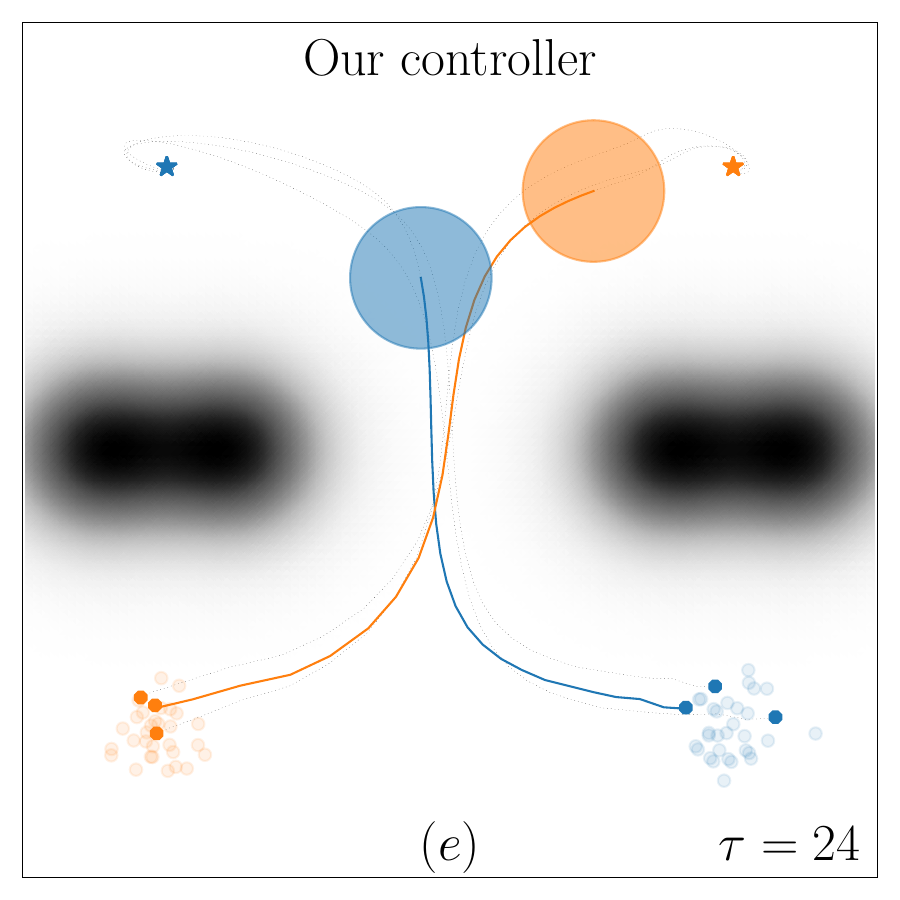}
	\end{minipage}%
	\begin{minipage}{0.3\linewidth}
        \includegraphics[width=\linewidth]{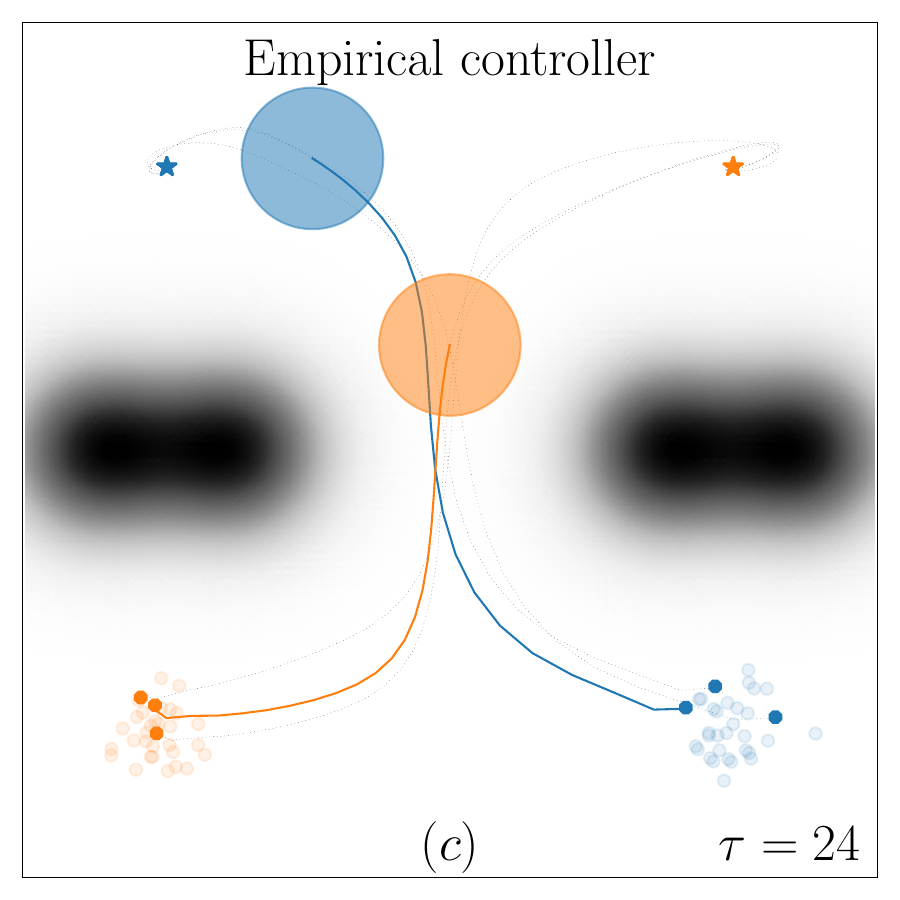}
	    \includegraphics[width=\linewidth]{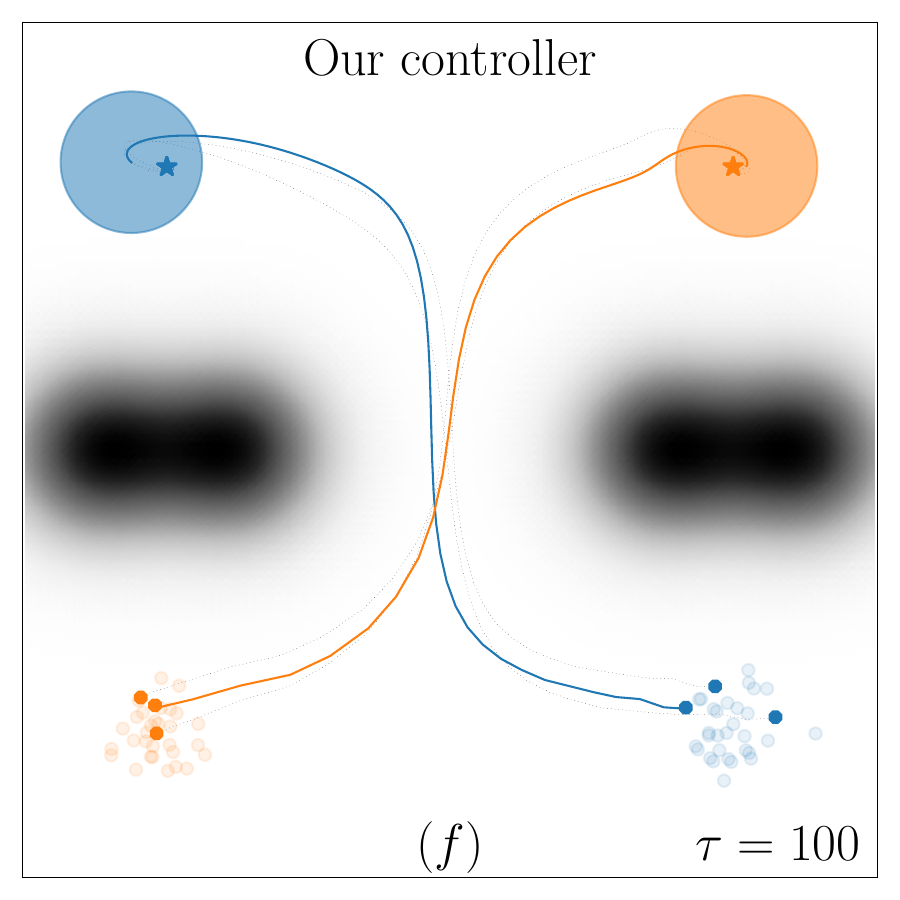}
	\end{minipage}
	\caption{
    Closed-loop test trajectories in the interval $[0,400]$ of
    (\textit{a}) the pre-stabilized system before training;
    (\textit{b})-(\textit{c}) the trained empirical controller; and
    (\textit{d})-(\textit{f}) one sampled controller from $\Q^*$ using SVGD.
    Training initial conditions are marked with $\circ$. 
    Snapshots are taken at instants $\tau$ indicated in each plot. 
    Colored balls represent the agents, with their radius indicating their size for collision avoidance.
    } 
	\label{fig:corridor}
\end{figure}

\paragraph*{\textbf{Scalability}} 
Using our approach, we train a controller with $n_\xi = 32$, $n_\zeta = 32$, leading to a total of $d = 11040$ parameters. The resulting test cost is $21.25$ and the collision ratio is $1.0\%$, improving the results in Table~\ref{tab:robots}. This experiment showcases the capability of our method in effectively training large controllers.

%%%%%%%%%%%%%%%%%%%%%%%%%%%%%%%%%%%%%%%%%%%%%%%%%%%%%%%%%%%%%%%%%%%%%%%%%%%%%%%%
\section{Conclusions and future work}
This paper presents a novel generalization bound for the SNOC cost, drawing upon randomized bounds within the PAC-Bayes theory.
By combining the derived bound with the unconstrained parametrization of stabilizing controllers proposed in~\cite{NeurSLS, furieri2024learning}, we propose an innovative
control policy design algorithm suitable for addressing SNOC problems.

Possible future directions are:
further analyzing the bound's asymptotic behavior,
exploring model selection techniques among a set of sampled controllers while maintaining valid generalization bounds, 
and utilizing the knowledge obtained from existing controllers to set the prior distribution.

% \addtolength{\textheight}{-3cm}   % This command serves to balance the column lengths

\bibliographystyle{IEEEtran}
\bibliography{root.bib}

\begin{thebibliography}{10}
\providecommand{\url}[1]{#1}
\csname url@samestyle\endcsname
\providecommand{\newblock}{\relax}
\providecommand{\bibinfo}[2]{#2}
\providecommand{\BIBentrySTDinterwordspacing}{\spaceskip=0pt\relax}
\providecommand{\BIBentryALTinterwordstretchfactor}{4}
\providecommand{\BIBentryALTinterwordspacing}{\spaceskip=\fontdimen2\font plus
\BIBentryALTinterwordstretchfactor\fontdimen3\font minus \fontdimen4\font\relax}
\providecommand{\BIBforeignlanguage}[2]{{%
\expandafter\ifx\csname l@#1\endcsname\relax
\typeout{** WARNING: IEEEtran.bst: No hyphenation pattern has been}%
\typeout{** loaded for the language `#1'. Using the pattern for}%
\typeout{** the default language instead.}%
\else
\language=\csname l@#1\endcsname
\fi
#2}}
\providecommand{\BIBdecl}{\relax}
\BIBdecl

\bibitem{userfriendly}
P.~Alquier, ``User-friendly introduction to {PAC}-{B}ayes bounds,'' \emph{Foundations and Trends{\textregistered} in Machine Learning}, vol.~17, no.~2, pp. 174--303, 2024.

\bibitem{majumdar2021pac}
A.~Majumdar, A.~Farid, and A.~Sonar, ``{PAC}-{B}ayes control: learning policies that provably generalize to novel environments,'' \emph{The International Journal of Robotics Research}, vol.~40, no. 2-3, pp. 574--593, 2021.

\bibitem{J1}
M.~Fard and J.~Pineau, ``{PAC-B}ayesian model selection for reinforcement learning,'' in \emph{Advances in Neural Information Processing Systems}, vol.~23, 2010.

\bibitem{J2}
M.~Fard, J.~Pineau, and C.~Szepesv{\'a}ri, ``{PAC-B}ayesian policy evaluation for reinforcement learning,'' \emph{Proceedings of the 27th Conference on Uncertainty in Artificial Intelligence, UAI 2011}, p. 195 – 202, 2011.

\bibitem{dawson2023safe}
C.~Dawson, S.~Gao, and C.~Fan, ``Safe control with learned certificates: A survey of neural {L}yapunov, barrier, and contraction methods for robotics and control,'' \emph{IEEE Transactions on Robotics}, vol.~39, no.~3, pp. 1749--1767, 2023.

\bibitem{schwan2023stability}
R.~Schwan, C.~N. Jones, and D.~Kuhn, ``Stability verification of neural network controllers using mixed-integer programming,'' \emph{IEEE Transactions on Automatic Control}, vol.~68, no.~12, pp. 7514--7529, 2023.

\bibitem{NeurSLS}
L.~Furieri, C.~L. Galimberti, and G.~Ferrari-Trecate, ``Neural system level synthesis: Learning over all stabilizing policies for nonlinear systems,'' in \emph{61st Conference on Decision and Control (CDC)}, 2022, pp. 2765--2770.

\bibitem{furieri2024learning}
------, ``Learning to boost the performance of stable nonlinear systems,'' \emph{IEEE Open Journal of Control Systems}, p. 342–357, 2024.

\bibitem{wang2021learning}
R.~Wang, N.~H. Barbara, M.~Revay, and I.~Manchester, ``Learning over all stabilizing nonlinear controllers for a partially-observed linear system,'' \emph{IEEE Control Systems Letters}, vol.~7, pp. 91--96, 2022.

\bibitem{furieri2022distributed}
L.~Furieri, C.~L. Galimberti, M.~Zakwan, and G.~Ferrari-Trecate, ``Distributed neural network control with dependability guarantees: a compositional port-{H}amiltonian approach,'' in \emph{Learning for Dynamics and Control Conference}.\hskip 1em plus 0.5em minus 0.4em\relax PMLR, 2022, pp. 571--583.

\bibitem{gu2021recurrent}
F.~Gu, H.~Yin, L.~El~Ghaoui, M.~Arcak, P.~Seiler, and M.~Jin, ``Recurrent neural network controllers synthesis with stability guarantees for partially observed systems,'' in \emph{AAAI}, 2022, pp. 5385--5394.

\bibitem{pauli2021offset}
P.~Pauli, J.~K{\"o}hler, J.~Berberich, A.~Koch, and F.~Allg{\"o}wer, ``Offset-free setpoint tracking using neural network controllers,'' in \emph{Learning for Dynamics and Control}.\hskip 1em plus 0.5em minus 0.4em\relax PMLR, 2021, pp. 992--1003.

\bibitem{bonassi2022recurrent}
F.~Bonassi, M.~Farina, J.~Xie, and R.~Scattolini, ``On recurrent neural networks for learning-based control: recent results and ideas for future developments,'' \emph{Journal of Process Control}, vol. 114, pp. 92--104, 2022.

\bibitem{abate2020formal}
A.~Abate, D.~Ahmed, M.~Giacobbe, and A.~Peruffo, ``Formal synthesis of {L}yapunov neural networks,'' \emph{IEEE Control Systems Letters}, vol.~5, no.~3, pp. 773--778, 2020.

\bibitem{berkenkamp2018safe}
F.~Berkenkamp, M.~Turchetta, A.~P. Schoellig, and A.~Krause, ``Safe model-based reinforcement learning with stability guarantees,'' \emph{Advances in Neural Information Processing Systems}, pp. 909--919, 2018.

\bibitem{revay2023recurrent}
M.~Revay, R.~Wang, and I.~R. Manchester, ``Recurrent equilibrium networks: Flexible dynamic models with guaranteed stability and robustness,'' \emph{IEEE Transactions on Automatic Control}, 2023.

\bibitem{SVGD}
Q.~Liu and D.~Wang, ``Stein variational gradient descent: A general purpose {B}ayesian inference algorithm,'' in \emph{Advances in Neural Information Processing Systems}, vol.~29, 2016.

\bibitem{Catoni}
O.~Catoni, ``P{AC-B}ayesian supervised classification: The thermodynamics of statistical learning,'' \emph{{IMS} Lecture Notes Monograph Series}, vol.~56, pp. 1--163, 2007.

\bibitem{Economou_Morari_nlIMC_1986}
C.~G. Economou, M.~Morari, and B.~O. Palsson, ``Internal model control: extension to nonlinear system,'' \emph{Industrial \& Engineering Chemistry Process Design and Development}, vol.~25, no.~2, pp. 403--411, 1986.

\bibitem{concentration}
S.~Boucheron, G.~Lugosi, and P.~Massart, \emph{{Concentration Inequalities: A Nonasymptotic Theory of Independence}}.\hskip 1em plus 0.5em minus 0.4em\relax Oxford University Press, 2013.

\bibitem{green2015bayesian}
P.~J. Green, K.~{\polishL}atuszy{\'n}ski, M.~Pereyra, and C.~P. Robert, ``Bayesian computation: a perspective on the current state, and sampling backward and forwards,'' \emph{Stat. Comput}, vol.~25, no.~4, pp. 835--862, 2015.

\bibitem{Alquier2016variational}
P.~Alquier, J.~Ridgway, and N.~Chopin, ``On the properties of variational approximations of {G}ibbs posteriors,'' \emph{Journal of Machine Learning Research}, vol.~17, no. 236, pp. 1--41, 2016.

\bibitem{pacoh}
J.~Rothfuss, V.~Fortuin, M.~Josifoski, and A.~Krause, ``{PACOH}: {B}ayes-optimal meta-learning with {PAC}-guarantees,'' in \emph{International Conference on Machine Learning}.\hskip 1em plus 0.5em minus 0.4em\relax PMLR, 2021, pp. 9116--9126.

\bibitem{pacpfl}
M.~G. Boroujeni, A.~Krause, and G.~Ferrari-Trecate, ``Personalized federated learning of probabilistic models: A {PAC}-{B}ayesian approach,'' \emph{arXiv preprint arXiv:2401.08351}, 2024.

\end{thebibliography}

\end{document}